\documentclass[review,onefignum,onetabnum]{siamart171218}

\usepackage{lipsum}
\usepackage{amsfonts}
\usepackage{graphicx}
\usepackage{epstopdf}
\usepackage{algorithmic}
\usepackage{booktabs}

\usepackage{color}

\usepackage{cite}
\ifpdf
  \DeclareGraphicsExtensions{.eps,.pdf,.png,.jpg}
\else
  \DeclareGraphicsExtensions{.eps}
\fi

\newsiamremark{remark}{Remark}
\newsiamremark{hypothesis}{Hypothesis}
\crefname{hypothesis}{Hypothesis}{Hypotheses}
\newsiamthm{claim}{Claim}

\headers{Efficient algorithms for computing FRC}{D. Barros de Souza {\it et al.}}

\title{Efficient set-theoretic algorithms for computing high-order Forman-Ricci curvature on abstract simplicial complexes\thanks{
This manuscript is for review purposes only.
\funding{This research is supported by the Basque Government through the BERC 2022-2025 program and by the Ministry of Science and Innovation: BCAM Severo Ochoa accreditation CEX2021-001142-S / MICIN / AEI / 10.13039/501100011033. Moreover, this research is financially supported by the IKUR Strategy under the collaboration agreement between the Ikerbasque Foundation and BCAM on behalf of the Department of Education of the Basque Government}}}

\author{Danillo Barros de Souza\thanks{Basque Center for Applied Mathematics 
  (\email{danillo.dbs16@gmail.com}, \email{dbarros@bcamath.org}, \url{http://www.bcamath.org/en/people/dbarros}).}
  \and Jonatas T. S. da Cunha \thanks{Universidade Federal de Pernambuco, Recife, Brazil, (\email{jteodomirosc@gmail.com}, \email{jonatas.teodomiro@ufpe.br})}
\and Fernando A. N. Santos \thanks{Dutch Institute for Emergent Phenomena (DIEP), Institute for Advanced Studies, University of Amsterdam, Oude Turfmarkt 147, 1012 GC, Amsterdam, The Netherlands, and Korteweg de Vries Institute for Mathematics, University of Amsterdam, Science Park 105-107, 1098 XG Amsterdam, the Netherlands
  (\email{f.a.nobregasantos@uva.nl})}
  \and Jürgen Jost \thanks{Max Planck Institute for Mathematics in the Sciences, Leipzig,  and Center for Scalable Data Analytics and Artificial Intelligence, Leipzig University, Germany,  and Santa Fe Institute, New Mexico, USA (\email{jost@mis.mpg.de})}
   \and Serafim Rodrigues \thanks{Basque Center for Applied Mathematics, Bilbao, Spain (\email{srodrigues@bcamath.org})}}

\usepackage{amsopn}

\makeatletter
\newcommand*{\addFileDependency}[1]{
  \typeout{(#1)}
  \@addtofilelist{#1}
  \IfFileExists{#1}{}{\typeout{No file #1.}}
}
\makeatother

\newtheorem{exmp}{Example}[section]
\ifpdf
\hypersetup{
  pdftitle={Efficient set-theoretic algorithms for high-order Forman-Ricci curvature on abstract simplicial complexes},
  pdfauthor={D. Barros de Souza, J. T. S. da Cunha, F. A. N. Santos Jürgen Jost and S. Rodrigues}
}
\fi
\nolinenumbers
\usepackage{color,soul}

\begin{document}
\maketitle

\begin{abstract}
Forman-Ricci curvature (FRC) is a potent and powerful tool for analysing empirical networks, as the distribution of the curvature values can identify structural information that is not readily detected by other geometrical methods. Crucially, FRC captures higher-order structural information of clique complexes of a graph or Vietoris-Rips complexes, which is not readily accessible to alternative methods. However, existing FRC platforms are prohibitively computationally expensive. Therefore, herein we develop an efficient set-theoretic formulation for computing such high-order FRC in simplicial complexes. Significantly, our set theory representation reveals previous computational bottlenecks and also accelerates the computation of FRC. Finally, We provide a pseudo-code, a software implementation coined FastForman, as well as a benchmark comparison with alternative implementations. We envisage that FastForman will be used in Topological and Geometrical Data analysis for high-dimensional complex data sets. Moreover, our development paves the way for future generalisations towards efficient computations of FRC on cell complexes.
\end{abstract}



\section{Introduction}
Network analysis is one of the success stories of complex systems research. The basic is simple, to represent the network as a graph whose mathematical features can then be studied. Beyond the determination of clusters or the identification of particularly important vertices or edges, also certain statistical features can reveal valuable structural information. Among those, the so-called discrete Ricci curvatures (a name derived from their conceptual origin in Riemannian geometry) are particularly useful. And among those curvatures, the Forman Ricci curvature (FRC) is particularly easy to compute. But this comes at the expense of ignoring valuable information that is contained in higher-order relations between the vertices.  To uncover such information, it is necessary to go to simplicial complexes, in the case at hand the clique complex of the graph. The topological aspects of such higher-order information are evaluated in Topological Data Analysis (TDA). There, with a scale parameter $r$, one constructs a graph by connecting pairs of data points whose distance is $\le r$, and one then computes the homology, and in particular, the Betti numbers,  of the resulting clique complex, called the Vietoris-Rips complex, as depending on $r$. This by now is a well-established method. But in line with the preceding, we want to extract geometric information, and as suggested, this should be done by the augmented FRC of such a complex. 
Now, these clique or Vietoris-Rips complexes are not arbitrary simplicial complexes but enjoy the special property that every simplex is automatically filled as soon as all its edges, that is, its 1-skeleton are present. \\
However, while the FRC of a graph is very easy to compute, this is no longer so for the augmented FRC of a simplicial complex, at least with standard methods. Therefore, in this paper, we provide a systematic mathematical approach within which the FRC of such a Vietoris-Rips type complex can be easily and quickly computed.
 As the name indicates, FRC was originally introduced by Forman \cite{forman2003bochner}. It was inspired by Bochner's formula for Ricci tensors in Riemannian manifolds \cite{bochner1946vector}. In fact, \cite{forman2003bochner} provides a formula for FRC of a cell complex  (also called CW complex). 
  This approach yields detailed structural insights surpassing those offered by traditional topological data analysis techniques such as persistent diagrams. Despite the broader theoretical applicability of FRC on CW complexes, our approach using set theory specifically applies to computing FRC on simplicial complexes. This restricts the domain of validity of our work, however,  this choice is still ubiquitous across various disciplines due to its simplicity and feasibility and also remains critically important for in-depth topological data analysis. 
 The graph version has already found a wide range of applications. For instance, it has been applied in cancer diagnostics \cite{sandhu2015graph}, detecting anomalies on brain networks \cite{chatterjee2021detecting}, detection of dynamic changes on data sets \cite{weber2016forman}, detection of stock market fragility \cite{sandhu2016ricci} and to infer pandemic states \cite{Barros_de_Souza_2021}.
 But in those and other cases, usually only the FRC of a graph was computed. To reveal the full power of the scheme, and to systematically compare the results with those obtained by TDA, we need to compute the FRC of such simplicial complexes.
But in general, the time and memory complexity of computing higher-order dimensional information from complex networks is prohibitively expensive, which often limits applications. For example, the time for computing all cliques in a graph has exponential-time complexity \cite{dieu1986average,fellows2009clique,peeters2003maximum}.
Also, the Betti number computation required for TDA is an NP-hard problem \cite{edelsbrunner2014computational}, due to the matrix construction of boundary operators. Indeed, Betti number computation increases with the number of cliques and, consequently, requires memory-efficient algorithms. Likewise, the lack of efficient algorithms for computing generalized geometric structures has limited the use of high-order Forman-Ricci curvature (FRC) for simplicial complexes.
Set-theoretic approaches \cite{enderton1977elements,jech2003set,zimmermann1985applications,pawlak2002rough} may offer a solution to this problem, and this is the approach of the present paper. 
Despite not being commonly used, set-theoretical approaches were first applied to surfaces and posets \cite{bloch2014combinatorial} and further to simplicial complexes \cite{yadav2022poset}.
To implement and apply this, we make an important observation that when the complex is of Vietoris-Rips type, that is, all simplices are filled when their 1-skeleton is present, this leads to considerable computational simplifications. And, incidentally, this is also a reason why   TDA is such a powerful scheme in applications. Therefore, our assumption seems natural, and we then want to build an algorithm upon it. To this end, we developed a set-theoretic representation theory for higher-order network structures based on node neighbourhoods to give an alternative (and optimal) definition for FRC in terms of the classic node neighbourhood of a graph. This leads us to a representation theory that reveals previous computational bottlenecks and enables us to develop an efficient algorithm that markedly accelerates the computation of FRC in higher-order networks.
{We have implemented our algorithm - FastForman (see in \cite{kaggle_FRC_report}) - and benchmarked it with other implementations, namely GeneralisedFormanRicci \cite{generalisedformanricci} and HodgeLaplacians \cite{hodgelaplacians}. We found that our implementation boosted time processing and drastically reduced memory consumption during the process of computing FRC for low and high-dimensional network structures. These results reinforce our beliefs in the power of set-theoretic approaches to improve the efficiency of algorithm implementations for high-order geometric approaches to big data sets.} Moreover, this new formulation paves the way for the extension of the set-theoretical formulation of FRC for cell complexes, which will be crucial for the investigation of intrinsic phenomena that are present in real datasets and that are not properly explored in simplicial complexes.
%

\section{Network and higher-order network fundamentals}\label{sec:network_background}
An \textit{undirected simple graph} (or network) $G=(V,E)$ is defined by a finite set of nodes $V$ and \textit{a set of edges} $E=\{(x,y)\,|\, x,y\in V, x\neq y\}$. The \cref{ex:simple_graph} depicts the definition of a simple graph. 

Let $x,y \in V$. We say that $y$ is a \textit{neighbour} of $x$ if the edge $(x,y) \in E$. The \emph{neighbourhood} of the node $x$ is defined by a set of nodes that are connected to $x$ via an edge in $E$, and we denote by $\pi_x$. Formally,
\begin{equation}\label{eq:node_neighborhood}
    \pi_x=\{y\in V\,|\, (x,y) \in E \},
\end{equation}
See \cref{ex:node_neigh} for elucidation.
We will define the \textit{Graph neighbourhood} as the sequence of node neighbourhoods of a graph $G$. Formally,
\begin{equation}\label{eq:Graph_neigh}
    \text{Neigh(G)}=(\pi_x\,|\,x\in V)
\end{equation}
An easy and efficient implementation for finding node neighbourhood is provided in \cref{alg:Neig}.
\begin{algorithm}
\caption{Compute Graph Neighborhood{: \text{Neigh}(G)}}
\label{alg:Neig}
\begin{algorithmic}
\STATE{Input: $G=(V,E)$}
\STATE{$\pi_x=\emptyset,\forall x \in V$}
\FOR{$e= (x,y)\in E$}
\STATE{$\pi_x := \pi_x \cup \{y\}$}
\STATE{$\pi_y := \pi_y \cup \{x\}$}
\ENDFOR
\RETURN $\pi_x$ for all $x \in V$
\end{algorithmic}
\end{algorithm}
\subsection{Abstract simplicial complex}
The concepts of this section are as defined from \cite{zomorodian2005topology,zomorodian2010fast,edelsbrunner2022computational}. Let $C$  be a set together with a collection $\mathcal{S}.$  We say that $C$ is an abstract simplicial complex  if the following conditions are satisfied:
\begin{enumerate}
    \item For each $v \in C$, $\{v\} \in \mathcal{S};$
    \item If $\gamma \subseteq \alpha$ and  $\alpha \in \mathcal{S}$, then $\gamma \in \mathcal{S}$.
\end{enumerate}
This definition follows a geometrical perspective and is employed in a variety of works\cite{samal2018comparative,saucan2019discrete,saucan2021simple,yadav2023discrete}. For convenience, we define these elements with a notation which is in harmony with discrete geometric approaches applied to simplicial complexes. To this end, we denote \textit{faces} to refer to the elements that build the simplicial complex $\mathcal{S}$.
We say that $\alpha$ are d-faces (or equivalently $(d+1)$-simplex - as typically defined in topological data analysis) if $|\alpha|=d+1.$ When $\gamma\subseteq \alpha$, we say that $\gamma$ is a \textit{coface} of $\alpha$. In our approach, in the special cases where $|\alpha|$=$|\gamma|$+1, we will denote $\gamma<\alpha$, or alternatively, $\alpha>\gamma$. We denote the set of all $d$-faces by $C_d$, and we can clearly reconstruct $C$ by taking $C=\bigcup\limits^{}_{d} C_{d}.$ 
Let $G=(V,E)$ be a simple undirected graph. The collection $\mathcal{C}$ of all subsets of $\{x_0,x_1,\hdots,x_d\}\subseteq V$ such that the vertices span a face of $C$ is called a \textit{vertex scheme} of $C$. Clearly, $\mathcal{C}$ is a simplicial complex and can be identified to $C$ via isomorphism \cite{munkres2018elements}. This result allows us to identify a face $\alpha \in C_d$ with the a subset $\{x_0,\hdots,x_d\}\subseteq V.$ Also, any combinations of subsets with size $k$ span a $(k-1)$-face in $C_{k-1}$, for $k\in \{1,\hdots,d\}.$
We refer the reader to \cref{fig:CW_example} and accompanying \cref{ex:simplicial_complex} for a graphic exposition of the concepts of simplicial complex and $d$-faces.
\begin{figure}[htb!]
    \centering
    \includegraphics[width=\linewidth]{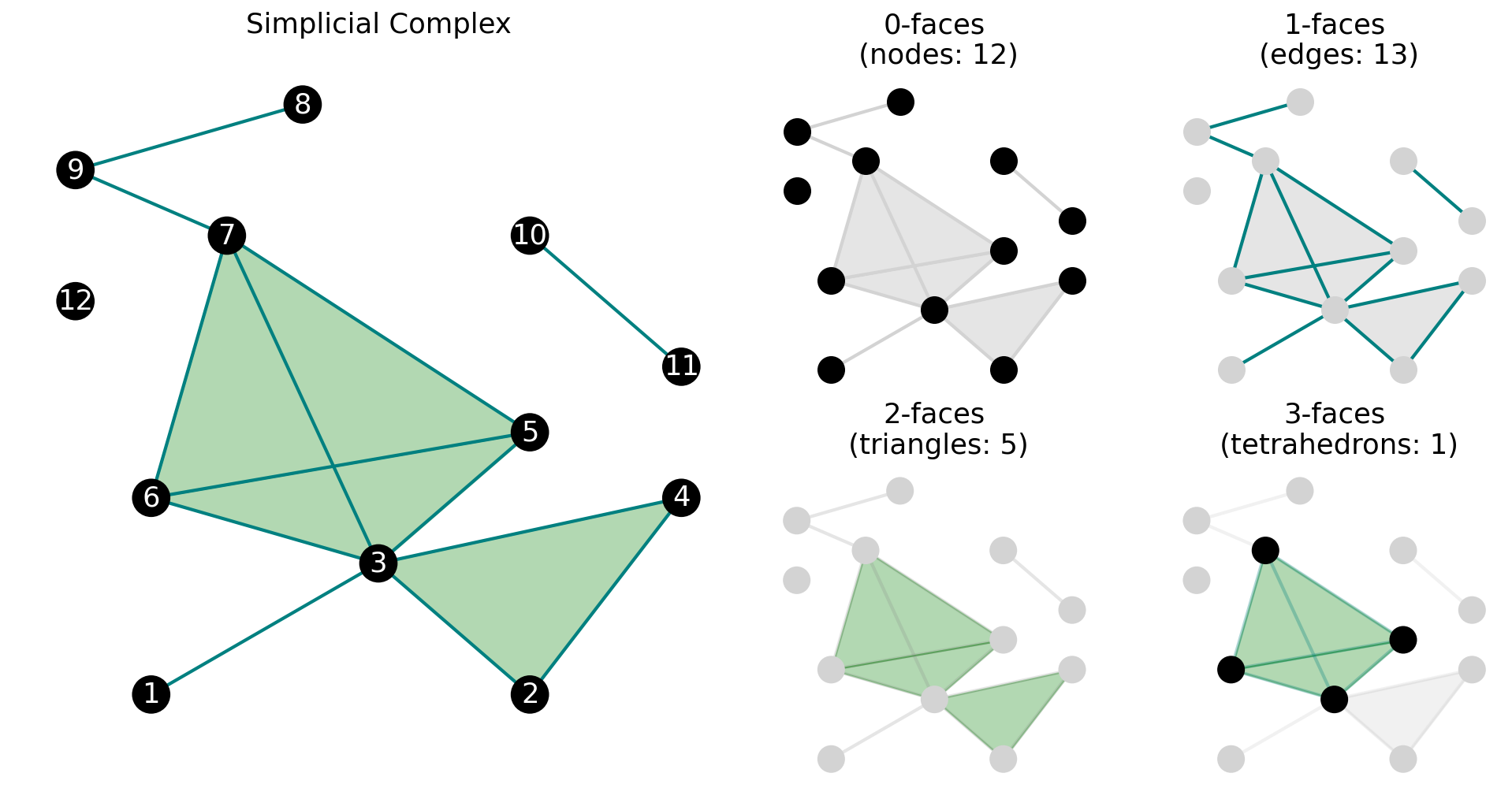}
    \caption{Example of a simplicial complex and its faces for $d\in \{0,1,2,3\}$.  In this example, we have $12$ nodes ($0$-faces), $13$ edges ($1$-faces), $5$ triangles ($2$-faces) and $1$ tetrahedron ($3$-faces).}
    \label{fig:CW_example}
\end{figure}
Analogous to definition \eqref{eq:node_neighborhood}, we define  the \textit{node} \textit{neighborhood of a $d$-face} $\alpha$ (or 
 the \textit{neighbouring of nodes of a $d$-face}) as follows: 
\begin{equation}\label{eq:cell_neighborhood}
    \pi_\alpha=\bigcap_{x \in \alpha}\pi_x.
\end{equation}
An application of this formulation can be accessed in \cref{ex:node_neigh_of_cell}.
Note that for $d=1$ in equation \eqref{eq:cell_neighborhood} we reach back to the classical formula of the neighbourhood of a node described by equation \eqref{eq:node_neighborhood}.
%
%
%
%
%
\section{High-order Forman-Ricci curvature}\label{sec:Forman_ricci_curvature}
The \textit{combinatorial Forman-Ricci curvature} (FRC) is originally defined in a more general formulation for cell complexes \cite{weber2017characterizing,saucan2018forman}, where the concept of neighbourhood and the parallelism is broader than in the context of simplicial complexes. Therefore, in this section, we develop all the concepts that are needed to define curvature restricted to simplicial complexes.
Let $G=(V,E)$ a graph and $d \in \mathbb{N}\bigcup \{0\}$.
Let  $\alpha \in C_{d}$. 
We define the \textit{boundary} of $\alpha$ by $\partial (\alpha)=\{\gamma \in C_{d-1},\,\gamma< \alpha\}$, for $d>0$ and $\partial(\alpha)=\emptyset$ for $d=0$.
Alternatively, we can define $\partial(\alpha)=\{\{x_0,\hdots,\hat{x_i},\hdots,x_d\},x_i\in \alpha\}=\{\gamma\subset\alpha\,,\, |\gamma|=d\}$,  where $\hat{x_i}$ denotes the removal of element $x_i$.
Note that $\partial(\alpha)$ has exactly $\binom{d+1}{d}=d+1$ elements. In \cref{ex:boundary_of_cell} the concept of face boundary is elucidated.
We say that $\alpha_1,\alpha_2 \in C_{d}$ are \textit{neighbours} if at least one of the following \textit{face neighborhood conditions} is satisfied:
 
\begin{enumerate}
    \item There exists a ($d-1$)-face $\gamma$ such that $\gamma<\alpha_1,\alpha_2$;
     \item There exists a ($d+1$)-face $\beta$ such that $\alpha_1,\alpha_2<\beta$.
  \end{enumerate} 
    We say that $\alpha_1$ and $\alpha_2$ are \textit{parallel neighbours} if either condition (i) or (ii) is reached, but not both simultaneously, and we denote the parallelism by $\alpha_1//\alpha_2$. In case of both conditions, (i) and (ii) are satisfied, $\alpha_1$ and $\alpha_2$ are said \textit{transverse neighbours}. {This neighbourhood condition {is exactly the same as originally defined for} cell complexes. However, once restricted to simplicial complexes, it can be simplified {since it is sufficient for} two faces to have a common coface {and thus they can be neighbours. This observation is proven in \cref{sec:demonstrations} (i.e., it is sufficient to satisfy} condition (i) so we have $\alpha_1$ and $\alpha_2$ as neighbours). As a consequence, it is enough to decide if condition (ii) is reached, which leads to the fact that these faces are either transverse or parallel. More precisely, in the context of simplicial complexes, we can say that $\alpha_1$  and $\alpha_2$ are parallel neighbours if condition (i) is reached, but not (ii). Otherwise, if both (i) and (ii) are reached, they are transverse neighbours. To make the development of our work more intuitive, we will use this definition for parallel and transverse neighbours henceforth.} The \cref{ex:neigh_exemple} elucidates the concept of face neighbourhood.
Let $\alpha \in C_{d}.$ We denote the set of all neighbours of a face $\alpha$ by $N_{\alpha}$. The sets of parallel and transverse neighbours of $\alpha$ will be denoted by $P_{\alpha}$ and $T_{\alpha}$, respectively. We will also denote by $H_{\alpha}$ as the set of all $(d+1)$-faces that have $\alpha$ in its boundary, i.e, all $\beta \in C_{d+1}$, such that $\beta>\alpha$.
 From conditions $1.$ and $2.$,  it is clear that these sets are disjoint, and we can also write 
    \begin{eqnarray}
    \label{eq:all_cells}
       N_{\alpha}=P_{\alpha}\sqcup T_{\alpha},
    \end{eqnarray}
for all $d\geq1$ and all $\alpha \in C_d.$
\begin{remark}
    We emphasize that the concepts of node neighbourhood of a face and face neighbourhood are strongly related, however, they describe completely different objects. While the node neighbourhood of a face describes the set of nodes that have a common intersection with all nodes of the face in question (as defined in \eqref{eq:node_neighborhood} and \eqref{eq:cell_neighborhood}), the face neighbourhood are all cells in the simplicial complex that reach at least one of the face neighbourhood conditions. A comparison between \cref{ex:node_neigh} and \cref{ex:node_neigh_of_cell} can elucidate the explanation of this difference. The link between the two concepts occurs from the fact that the neighbours of a face $\alpha$ can be described in terms of the node neighbourhood of the cofaces in its boundary, which will be clarified along the construction of our set-theoretic definitions. 
\end{remark}
 The \textit{$d$-th Forman-Ricci curvature} (FRC) is restricted to simplicial complexes from \cite{weber2017characterizing,saucan2018forman} and then given as follows:
     \begin{eqnarray}\label{eq:forman_d_cells}
    \mathrm{F}_{d}(\alpha)=&\left|\{ \beta \in C_{d+1},\,\, \beta > \alpha\}\right| \nonumber \\ & +\left|\{\gamma \in C_{d-1}, \,\, \gamma<\alpha\}\right| \nonumber \\&   -\left|\{\alpha'\in C_{d},\,\,\alpha'//\alpha \}\right|,
\end{eqnarray}
where $|.|$ is the number of elements in the set (i.e. cardinality). Alternatively, it can be formulated as follows:

\begin{eqnarray}
\label{eq:forman_ricci_set_theory}
   \mathrm{F}_{d}(\alpha)=\left|H_{\alpha}\right|+(d+1)-\left|P_{\alpha}\right|.
\end{eqnarray}

Note, for example, when $d=1$, we have recovered the classic FRC of an edge $e \in E=C_1$ as
\begin{eqnarray}\label{eq:forman_curvature_1_cells}
    \mathrm{F}(e)=\left|\{t\in C_2,\,\, t > e\}\right| \,+2 \,   -\left|\{e'\in E,\,\,e'//e\}\right|,
\end{eqnarray}
where $t$ denotes a triangle ($2$-face).
The $d$-th FRC of a non-empty $C_d$ complex is given by the average of the Forman-Ricci curvatures by face, \textit{i.e.,}

\begin{eqnarray}
    \mathrm{F}_{d}(C)= \dfrac{1}{|C_d|}\sum_{\alpha\in C_d}   \mathrm{F}_{d}(\alpha).
\end{eqnarray}
The computation of FRC for dimensions $d=1$ and $d=2$ can be clarified in \cref{ex:FRC_calculation}.
\section{Results}\label{sec:results}
Subsequently, we propose our set-theoretic representation theory for the neighbourhood of faces and develop the associated efficient algorithm computing high-order FRC in simplicial complexes.
\subsection{Set-theoretic representation of Forman-Ricci curvature}
The set of neighbours of $\alpha$  can be redefined in terms of the node neighbourhood as follows (see \cref{theo:neighb_charact}):
\begin{eqnarray}\label{eq:N_set2}
   N_\alpha=\bigsqcup_{\gamma\in\partial(\alpha)}\,\bigcup_{\substack{x\in\pi_\gamma\neq \emptyset\\x\notin\alpha}}\{\gamma\cup\{x\}\},
\end{eqnarray}
where the $\bigsqcup$ operation represents the disjoint union of sets.
The \cref{theo:2implies1} also characterizes the neighbours of a $d$-face by finding faces that reach the face neighbouring condition 1., which means that it is sufficient for two $d$-faces to have a common boundary to be neighbours.
The transverse neighbours of $\alpha$ can also be characterized (see \cref{theo:set_of_transverse}), which lead as to the formula
\begin{eqnarray}\label{eq:T_set_theory}
T_\alpha=\bigsqcup_{\gamma\in\partial(\alpha)}\,\bigcup_{\substack{x\in\pi_\gamma\neq \emptyset \\ x\in \pi_{\alpha}\neq \emptyset\\x\notin\alpha}}\{\gamma\cup\{x\}\}.
\end{eqnarray}

From (\ref{eq:T_set_theory}) and (\ref{eq:all_cells}), we derive the set of parallel  faces from a complementary set of transverse cells, 
\begin{equation}\label{eq:P_set_theory}
    P_{\alpha}=\bigsqcup_{\gamma\in\partial(\alpha)}\,\bigcup_{\substack{x\in\pi_\gamma\neq \emptyset \\ x\notin \pi_{\alpha}\neq \emptyset\\x\notin\alpha}}\{\gamma\cup\{x\}\}.
\end{equation}
In practical terms, equations (\ref{eq:T_set_theory}) and (\ref{eq:P_set_theory}) suggest that the transverse and parallel faces can be obtained from (\ref{eq:N_set2}) by deciding if each neighbouring face of $\alpha$ is contained in a $(d+1)$-face or not.
The details are shown in \cref{theo:set_of_parallel}. Also, from \cref{theo:T_cardinality},
\begin{equation}\label{eq:all_neigh_formula}
        \left|T_{\alpha}\right| = (d+1) \left|H_{\alpha}\right|,
    \end{equation}
which together with \eqref{eq:forman_ricci_set_theory} provides
\begin{equation}\label{eq:FRC_algorithm_v1}
    \mathrm{F}(\alpha)=\frac{1}{(d+1)}|T_{\alpha}|+(d+1)-|P_{\alpha}|,
\end{equation}
which after expansion provides
\begin{equation}\label{eq:FRC_algorithm_v1_symbolic}
    \mathrm{F}(\alpha)=\frac{1}{(d+1)}\sum_{\gamma \in \partial(\alpha)}|\pi_{\alpha}\cap \pi_{\gamma}-\alpha|+(d+1)-\sum_{\gamma \in \partial(\gamma)}|\pi_{\gamma}-\pi_{\alpha}-\alpha|.
\end{equation}
%
Moreover, we can express the number of faces of dimension $d+1$ containing $\alpha$ (see proof in \cref{theo:set_of_highn}) as
\begin{eqnarray}
   \label{eq:cells_containing_alpha}
   H_{\alpha}=\bigcup_{x \in \pi_\alpha \neq \emptyset } \{\alpha\cup\{x\} \}.
\end{eqnarray}
It follows from equations \eqref{eq:all_cells}, \eqref{eq:N_set2} and \eqref{eq:T_set_theory} that
%
%
%

\begin{equation}\label{eq:forman_d_cells_alternative}
    \mathrm{F}_{d}(\alpha)=|\pi_{\alpha}|+(d+1)-\sum_{\gamma \in \partial(\gamma)}|\pi_{\gamma}-\pi_{\alpha}-\alpha|.
\end{equation}
%
%
Equation \eqref{eq:forman_d_cells_alternative} provides a direct computation of FRC from the parallel faces. However, an alternative equation can be derived by computing FRC directly from the face neighbourhood. This is obtained by combining \eqref{eq:all_cells}, \eqref{eq:forman_ricci_set_theory} and \cref{eq:all_neigh_formula} we have
%
\begin{eqnarray}\label{eq:forman_ricci_alternative_version}
   \mathrm{F}_{d}(\alpha)=(d+2)\left|\pi_\alpha\right|+(d+1)-\left|N_{\alpha}\right|.
\end{eqnarray}
The expansion of the term $\left|N_{\alpha}\right|$ leads to
\begin{equation}
\left|N_{\alpha}\right|=\sum_{\gamma \in \partial(\alpha)}|\pi_{\gamma}-\alpha|=\sum_{\gamma \in \partial(\alpha)}|\pi_\gamma|-(d+1),
\end{equation}
which when applied to \eqref{eq:forman_ricci_alternative_version}  leads us to
\begin{eqnarray}\label{eq:forman_ricci_algorithm_version}
    \mathrm{F}_{d}(\alpha)=(d+2)\left|\pi_{\alpha}\right|+2\cdot(d+1)-\sum_{\gamma \in \partial(\alpha)}\left| \pi_{\gamma} \right|,
\end{eqnarray}
%
%
Equations \eqref{eq:FRC_algorithm_v1_symbolic},\eqref{eq:forman_d_cells_alternative} and \eqref{eq:forman_ricci_algorithm_version} not only establish new formulations for FRC but also crucially enable us to propose novel algorithms to efficiently compute FRC in terms of the local node neighbourhoods of the faces into consideration. We propose \cref{alg:frc,alg:frc1,alg:frc2} that should be parameterized by the graph $G=(V,E),$ the list of faces ($C$), and the maximum face dimension ($d_{\max}$), and it computes FRC up to dimension $d\leq d_{\max}$.
\begin{algorithm}
\caption{Compute Forman-Ricci Curvature: Method A}
\label{alg:frc1}
\begin{algorithmic}
\STATE{Input: $C,d_{\max}\geq 1$}
\STATE{$\mathcal{F}_d=\emptyset,\forall d\leq d_{\max}$}
\STATE{Compute \text{Neigh(G)} from \cref{alg:Neig};}
\FOR{$\alpha\in C$}
\STATE{$d:=|\alpha|-1$}
\STATE{Compute $\mathrm{F}_{d}(c)$ according to \eqref{eq:FRC_algorithm_v1_symbolic}:

    \STATE{Compute $\pi_{\alpha};$}
    \FOR{ $\gamma \in \partial(\alpha)$}
    \STATE{set $M:=\{v\in V\,|\,v\in \pi_{\gamma}, \, v \notin \alpha\}$, $t:=0$, $p:=0;$}
    \FOR{ $v \in M$}
    \IF{ $v \in \pi_\alpha$}
    \STATE{$t:=t+1$}
    \ELSE
    \STATE{$p:=p+1$}
    \ENDIF
    \ENDFOR
    \ENDFOR
    \STATE{set $\mathrm{F}_{d}(\alpha)=\frac{1}{(d+1)}\cdot t +(d+1)-p$}

}
\STATE{Update $\mathcal{F}_d := \mathcal{F}_d\cup \{ \mathrm{F}_{d}(\alpha)\}$}
\ENDFOR
\RETURN $\mathcal{F}=\bigcup\limits_{d}\{\mathcal{F}_d\}$
\end{algorithmic}
\end{algorithm}
\begin{algorithm}
\caption{Compute Forman-Ricci Curvature: Method B}
\label{alg:frc2}
\begin{algorithmic}
\STATE{Input: $C,d_{\max}\geq 1$}
\STATE{$\mathcal{F}_d=\emptyset,\forall d\leq d_{\max}$}
\STATE{Compute \text{Neigh(G)} from \cref{alg:Neig};}
\FOR{$\alpha\in C$}
\STATE{$d:=|\alpha|-1$}
\STATE{Compute $\mathrm{F}_{d}(c)$ according to \eqref{eq:forman_d_cells_alternative}:

    \STATE{Compute $\pi_{\alpha};$}
    \STATE{Sset $p:=0;$}
    \FOR{ $\gamma \in \partial(\alpha)$}
    \STATE{set $M:=\{v\in V\,|\,v\in \pi_{\gamma}, \, v \notin \alpha\}$}
    \FOR{ $v \in M$}
    \IF{ $v \notin \pi_\alpha$}
    \STATE{$p:=p+1$}
    \ENDIF
    \ENDFOR
    \ENDFOR
    \STATE{set $\mathrm{F}_{d}(\alpha)=|\pi_{\alpha}| +(d+1)-p$}

}
\STATE{Update $\mathcal{F}_d := \mathcal{F}_d\cup \{ \mathrm{F}_{d}(\alpha)\}$}
\ENDFOR
\RETURN $\mathcal{F}=\bigcup\limits_{d}\{\mathcal{F}_d\}$
\end{algorithmic}
\end{algorithm}
\begin{algorithm}
\caption{Compute Forman-Ricci Curvature: Method C}
\label{alg:frc}
\begin{algorithmic}
\STATE{Input: $C,d_{\max}\geq 1$}
\STATE{$\mathcal{F}_d=\emptyset,\forall d\leq d_{\max}$}
\STATE{Compute \text{Neigh(G)} from \cref{alg:Neig};}
\FOR{$\alpha\in C$}
\STATE{$d:=|\alpha|-1$}
\STATE{Compute $\mathrm{F}_{d}(\alpha)$ according to \eqref{eq:forman_ricci_algorithm_version}:}
\STATE{Set $h:=|\pi_\alpha|$ and $n:=0;$}
\FOR{$\gamma \in \partial(\alpha)$}
\STATE{Update $n:=n+|\pi_{\gamma}|;$ }
\STATE{Set $\mathrm{F}_{d}(\alpha)=(d+2)h+2(d+1)-n;$}
\ENDFOR
\STATE{Update $\mathcal{F}_d := \mathcal{F}_d\cup \{ \mathrm{F}_{d}(\alpha)\}$}
\ENDFOR
\RETURN $\mathcal{F}=\bigcup\limits_{d}\{\mathcal{F}_d\}$
\end{algorithmic}
\end{algorithm}
\begin{figure}[!htb]
    \centering
    \includegraphics[width=\linewidth]{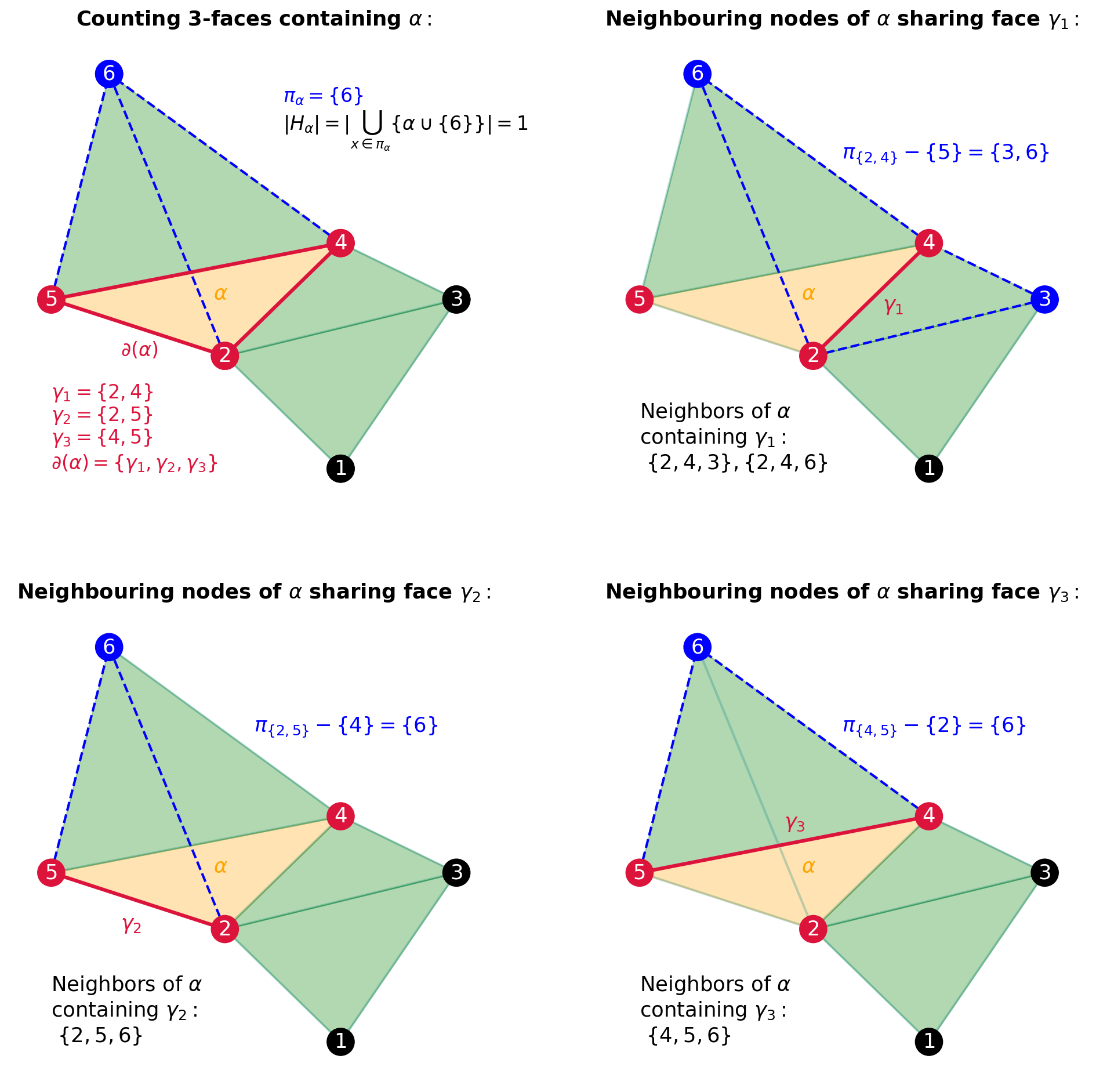}
    \caption{Example of how the algorithm detects face neighbours from node neighbourhood for $2$-faces (triangles). The $2$-face $\alpha$ is enhanced in orange, whilst its boundary is drawn in red. According to the algorithm, the number of $3$-faces containing $\alpha$ in its boundary ($|H_{\alpha}|$) coincides with the number of common neighbours of the nodes in $\alpha$ ($|\pi_\alpha|$). Also, the number of neighbours per boundary is counted by the size of intersections of node neighbourhoods in the boundary. In this example, $\mathrm{F}_2(\alpha)=4\cdot|\{6\}|+3-(|\{3,6\}|+|\{6\}|+|\{6\}|)=3.$ The same output is reached by using \eqref{eq:forman_d_cells}: $\mathrm{F}_2(\alpha)=1+3-1=3.$}  
    \label{fig:algorithm_example}
\end{figure}
\subsection{Computation and Benchmark of Novel Algorithms}
To validate our set-theoretic representation of FRC, we developed a Python language version \cite{python} of our pseudo-codes, which we named \texttt{FastForman} - Method A, B and C, which uses the formulation \eqref{eq:FRC_algorithm_v1_symbolic},  \eqref{eq:forman_d_cells_alternative} and  \eqref{eq:forman_ricci_algorithm_version}, respectively. Furthermore, we run benchmark tests and compared \texttt{FastForman}'s execution time and memory usage with the leading software in the literature, namely Python packages \texttt{HodgeLaplacians} \cite{hodgelaplacians} and \texttt{GeneralisedFormanRicci} \cite{generalisedformanricci}. The benchmark tests were performed on a Dell laptop model XPS 15-7590, Intel Core i7-9750H CPU 2.60GHz, 32Gb RAM, running Linux Ubuntu operating system version 20.04.4 LTS (Focal Fossa).
To carry out the tests, with the help of \texttt{NetworkX} package \cite{hagberg2020networkx}, we generated $20$ copies of $3$-dimensional point cloud data from random geometric networks \cite{penrose2003random} with the parameters $n=50$ nodes and radius $r\in \{0.1,0.2,\hdots,1.0\}$ (to access this data we refer the reader to the link \cite{kaggle_FRC_report}). To undertake an unbiased benchmark test, it is first worth noting that the Python package, \texttt{GeneralisedFormanRicci}, is limited to computing the FRC for edges and triangles. Therefore, we divided the benchmark tests into two separate groups. The first group aimed at comparing \texttt{GeneralisedFormanRicci}, \texttt{HodgeLaplacians} and \texttt{FastForman} (A, B and C) for $d_{max}=2$. The second group focused on comparing \texttt{HodgeLaplacians} versus \texttt{FastForman} (A, B and C) for $d_{max}=5$. To compare the time complexity between all involved software packages, we considered the total execution time from the generation of faces to the end computation of FRC. For a fair comparison, we employed the same software package, \texttt{Gudhi} \cite{maria2014gudhi}, to list the faces of the aforementioned generated networks.  Due to the separate benchmark groups, we first depict in \cref{fig:avg_num_cliques} the average number of faces for $d_{\max}\in\{2,5\}$ contained within the aforementioned generated networks. In \cref{fig:All_time_processing_new} we benchmark the average of the total processing time and average memory peak usage for computing the FRC over the generated random networks. We also compare the results with the time and memory consumption of listing the set of faces via the \texttt{Gudhi} package (light blue curves in the figure). We subsequently discuss the merits and limitations of each software package for computing FRC. 

\texttt{GeneralisedFormanRicci}: Due to computational limitations with memory management, we could not compute FRC for radius values higher than $r=0.7$ via GeneralisedFormanRicci. 
Despite GeneralisedFormanRicci being limited to lower dimensional computations, the time consumption was the highest, exceeding the \texttt{Hodgelaplacians} and \texttt{FastForman} implementations. Moreover, it displays a high variance of time and memory consumption particularly for dense networks, which might be associated with its limitations in computing FRC up to 2-dimensional faces (compare  \cref{fig:avg_num_cliques,fig:All_time_processing_new}). The limitation to point cloud data sets is also a drawback of the algorithm since it restricts one from selecting a subset of faces over which one may compute FRC.

\texttt{Hodgelaplacians:} The time processing is comparable with the performance results of \texttt{FastForman} for low-dimensional faces and low connectivity values. However, the computational efficiency decreases for high-order faces and larger radii and thus its performance falls behind in comparison to \texttt{FastForman}. Furthermore, the time and memory processing variance increase significantly with respect to the network complexity (\textit{ie.}, size, number of higher-order faces and density). This fact might be explained by the computation of FRC via Bochner's discrete formula, which uses Hodge Laplacian and Bochner Laplacian matrix operators \cite{borner2007network,weber2017characterizing,dakurah2022modelling}. This demands memory for storing the matrices and time to both generate the matrices and compute FRC. It is worth emphasising that the Hodge Laplacian matrix for the highest face dimension is simply the $d_{\max}$-th boundary operator, which in computational terms, induces a miscalculation on FRC for this dimension. Thus, for the correct computation for all faces, the entire simplicial complex should be provided, which impacts directly on the computational complexity. One way to improve efficiency and avoid miscalculation would be to include information of $(d_{\max}+1)$-faces so the correct computation up to $d_{\max}$ can be computed correctly. A detailed report of this special case is included in the link \cite{kaggle_FRC_report}, where we provide an example of the miscalculation.

\texttt{FastForman:} Our implementations had equivalent performance and superseded existing FRC software packages benchmarked in our tests. We provide our  Python implementation in \cite{Barros_de_Souza_FastForman_-_An_2024}
The key to this performance is that our set-theoretic formulation leads us to algorithms that only require the use of local node neighbourhood information for computing the FRC by face, which, as a consequence, reduces the complexity of computing FRC. Noteworthy, there is no requirement for specifically computing and storing neighbouring faces, hence this explains its low memory usage and fast computation. Moreover, \texttt{FastForman} implementations are versatile in the sense that they do not make mandatory the inclusion of the global information of an entire simplicial complex. In other words, the FRC computation is flexible, in a way that it can be computed for subsets of faces within a given simplicial complex. The limiting factors are twofold: Firstly, the complexity of the algorithms is dependent upon how faces 
are computed, which is generically an NP-problem and herein we employed the \texttt{Gudhi} algorithm. Secondly, the algorithm's complexity is associated with node neighbourhood computation, as given by \cref{alg:Neig}, which has a time complexity of $\mathcal{O}(|E|)$ since it iterates over edges. Indeed, \texttt{FastForman}'s implementation processing time increases in accordance with the time required for finding faces; see \cref{fig:All_time_processing_new,fig:avg_num_cliques}. \textcolor{black}{To be more precise, the space complexity is, in the worse case (when the graph is complete), $\mathcal{O}(\mathrm{C}(d_{\max})+|V|^2)$, where $\mathrm{C}(d_{\max})$ is the space complexity for the finding cliques algorithm (up to dimension $d_{\max}$) and $|V|^2$ is the space complexity for storing the neighbourhood of nodes, in the worse scenario. The time complexity is $\mathcal{O}(\sum_{d=1}^{d_{\max}}|\mathrm{C}_{d}|\cdot(d+1)\cdot |V|^{(d+1))}$, where $(d+1)$ is the time complexity for iterating the boundary of a $d$-face and $|V|^{(d+1)}$ is the order complexity for computing set intersection of $(d+1)$ node neighbourhoods (taking into account that we can compute set intersection in $\mathcal{O}(\sum_{x \in \alpha} |\pi_x|$) time.}
\textcolor{black}{As an application, we showcase our optimal performance in random graphs, aiming to extend the results on geometry detection from the work in \cite{weber2017characterizing}. From \cref{theo:FRC_range}, we are able to compare the exact range of values reached by FRC as a function of the number of nodes and the maximum face dimensions. Moreover, we show we can easily compute the frequency (or distribution) of FRC values in function of the dimension of the faces and compare them across different geometric graphs. To this end, we generated the expected frequency of FRC values for $1000$ random graphs in each of the following cases: Erd\H{o}s-Renyi, 2D Random Geometric, Barab{\'a}si-Albert and Watts-Strogatz model (with edge randomization parameter $p=0.25$). All graphs were set up with $n=50$ nodes and an approximate density of $0.5$. We computed the FRC distribuitions for $d\in \{1,2,3,4\}$, which can be seen in \cref{fig:FRC_frequency}. Observe that the standard deviation of the FRC values decreases as the FRC dimension increases. Also, the Watts Strogatz curve becomes less distinguishable from Erd\H{o}s-Renyi, which may indicate that the random features of Watts Strogatz can be detected only via higher-order analysis. In general, all FRC could be used for geometry detection.}

%
%
\begin{figure}[H]

    \centering
    \includegraphics[width=\linewidth]{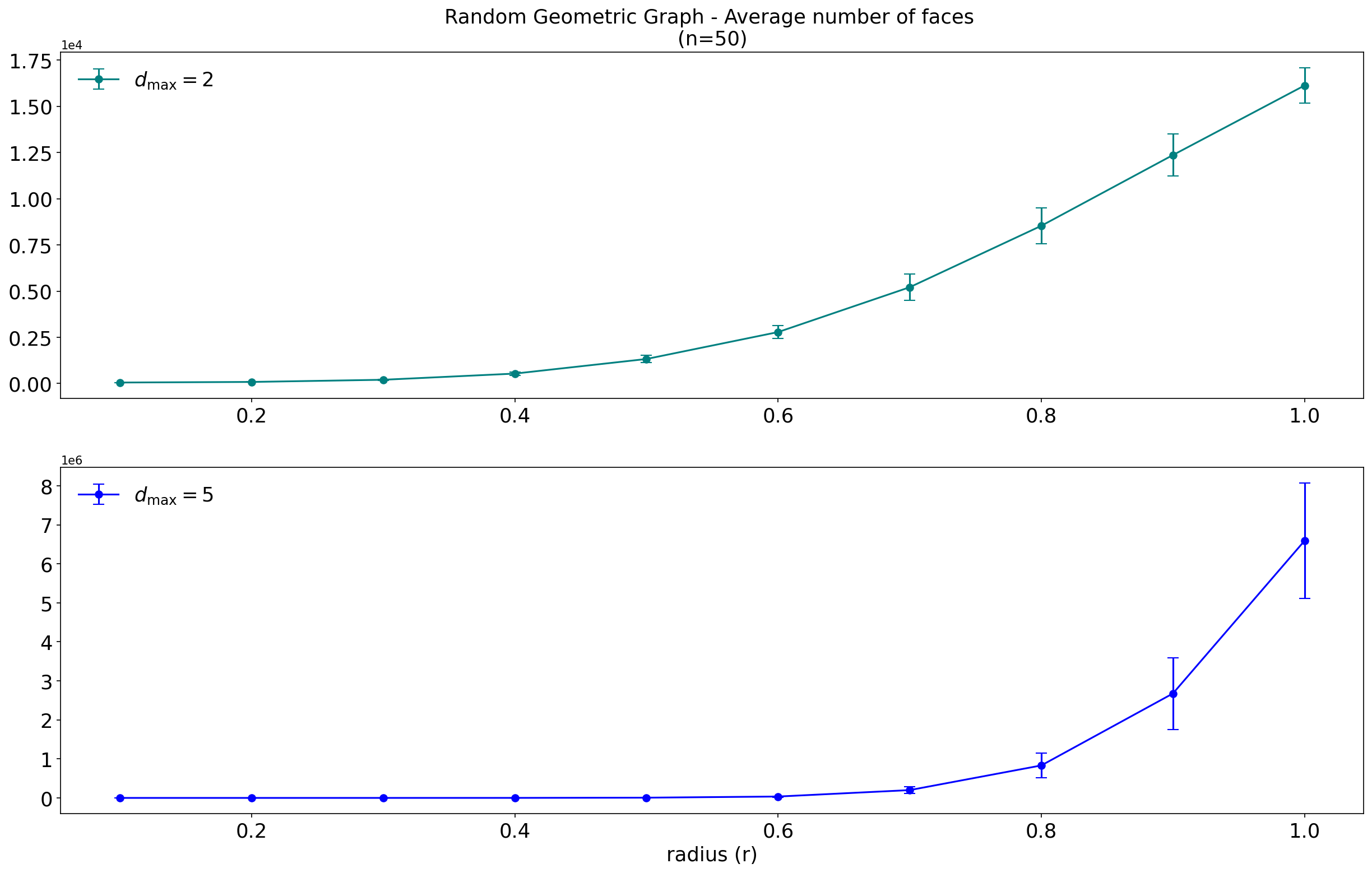}
    \caption{Average number of $d$-faces for the set of $50$ point cloud data, and $d_{\max}\in\{2,5\}.$}
    \label{fig:avg_num_cliques}
\end{figure}
\begin{figure}[H]
    \centering
    \includegraphics[width=\linewidth]{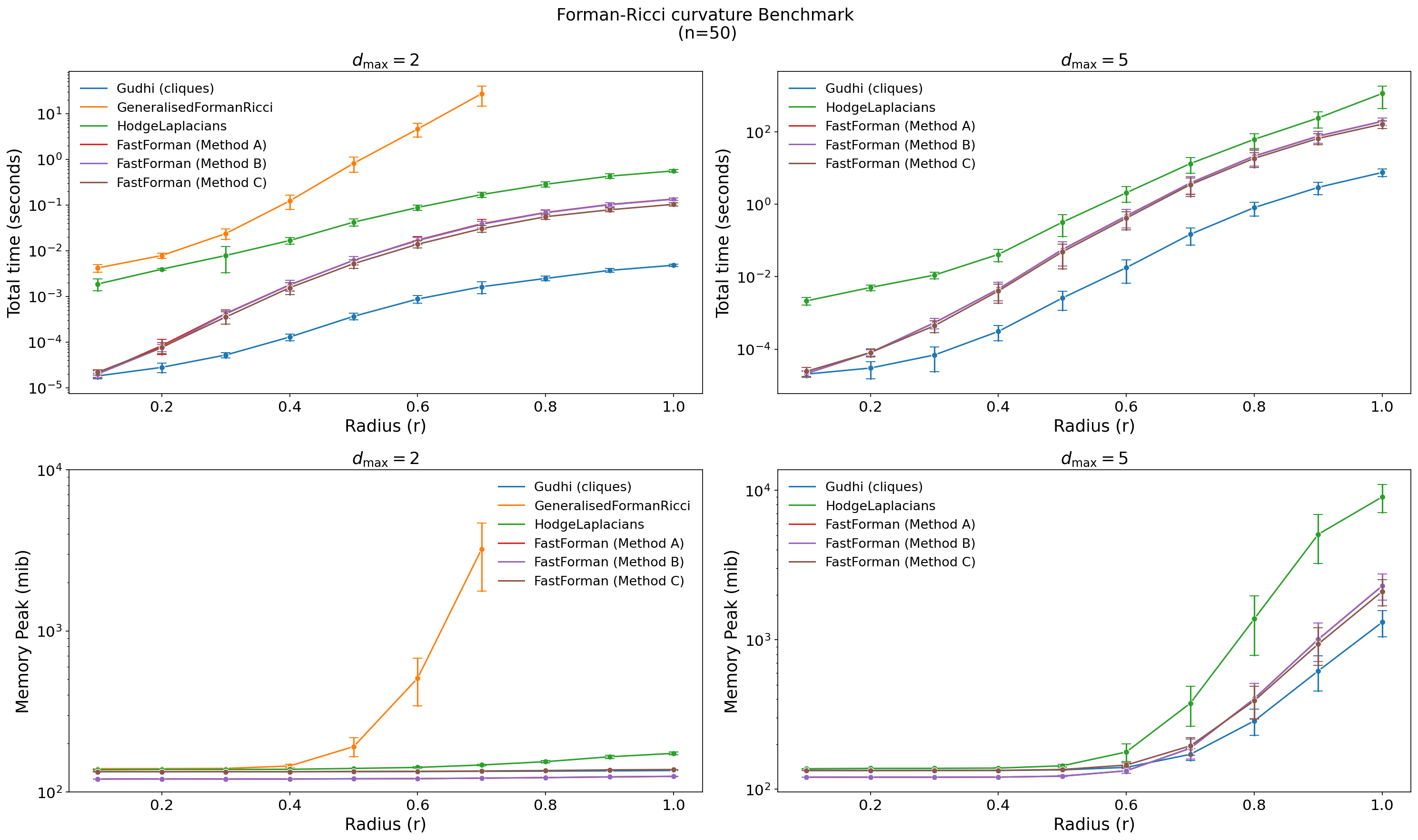}
    \caption{Benchmark for computing Forman-Ricci curvature by using HodgeLaplacians, GeneralisedFormanRicci and FastForman With methods A, B and C, in comparison with the faces computation by using the Gudhi algorithm, for $n=50$ and $d_{\max}\in\{2,5\}$. We display the result in a log scale to clarify the analysis.}
    \label{fig:All_time_processing_new}
\end{figure}

\begin{figure}[!htb]
    \centering
    \includegraphics[width=\linewidth]{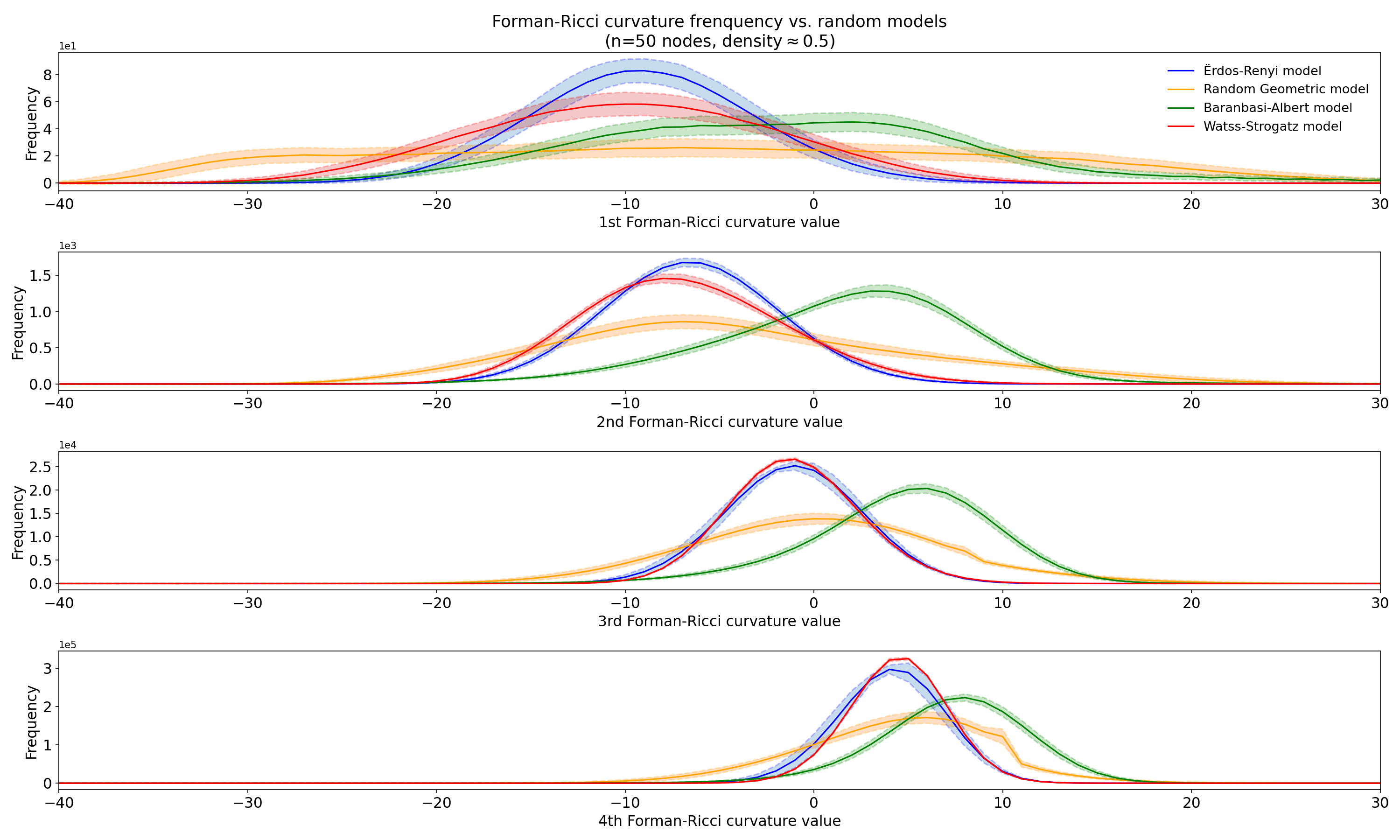}
    \caption{Comparison of FRC for higher order dimensions between different random graph models.}  
    \label{fig:FRC_frequency}
\end{figure}
\section{Conclusion}
\label{sec:conclusion}
We proposed novel set theoretical formulations for high-order Forman-Ricci curvature (FRC) {restricted to simplicial complexes,} and this led us to the implementations of new algorithms, which we coined \texttt{FastForman} (Method A, B and C). Central to our set-theoretic representation theory for high-order network cells and FRC computation is the fact we managed to redefine parallelism and transversality of neighbouring faces. This enabled us to formulate FRC computation in terms of the neighbours of nodes in a network. As a consequence, we managed to computationally tame some of the combinatorial complexities intrinsic to high-order networks. We validated \texttt{FastForman} algorithms via benchmark tests on random higher-order geometric graphs and compared them with leading softwares, namely, \texttt{Hodgelaplacians} and \texttt{GeneralisedFormanRicci}. To ensure a fair comparison, we implemented \texttt{FastForman} with Python language since \texttt{Hodgelaplacians} and \texttt{GeneralisedFormanRicci} are also implemented in Phyton. We also took into consideration the intrinsic limitations of \texttt{Hodgelaplacians} and \texttt{GeneralisedFormanRicci}, such as the maximum number of face dimensions tolerated by these algorithms. Our benchmark tests established that \texttt{FastForman} implementations are more efficient in time and memory complexity. In conclusion, despite the constraints of face search NP algorithms and the combinatorial explosion of high-order network structures, our findings facilitate the computation of higher-order geometric invariants in simplicial complexes. We envisage that \texttt{FastForman} will open novel research avenues in big data science, particularly through a geometrical point of view. Interestingly, our findings based on simplicial complexes motivate the development of novel efficient algorithms and constructive proofs beyond simplicial complexes and in particular for the geometric computations based on cell complexes. Such developments would constitute the starting point towards exploring real data through the lens of cell complexes, which is hardly explored in classical simplicial complexes approaches. We envisage that this progression will have major impact in big data science and complex systems.
\pagebreak
\section{Appendix}
In this section, we provide examples, as well as the demonstrations of the theory developed in our work.
\subsection{Examples}\label{sec:examples}
In this section, we provide examples of our work.
\begin{exmp}[Undirected Simple Graph]\label{ex:simple_graph}
    In \cref{fig:CW_example}, we have a undirected simple graph, $G=(V,E)$, such that $V=\{1,\hdots ,12\}$ and 
    \begin{eqnarray}
    \nonumber
        E=&\\  \{(1,3),(2,3),(2,4),(3,4),(3,5),(3,6),(3,7),(5,6),(5,7),(6,7),(7,9),(8,9),(10,11)\} \nonumber
    \end{eqnarray}
\end{exmp}
\begin{exmp}[Neighborhood of a node]\label{ex:node_neigh} In \cref{fig:CW_example}, we have the neighbours of each node in $V$ as described below:
    \begin{eqnarray}
        \pi_1=\{3\}, \nonumber\\ \nonumber
        \pi_2=\{3,4\}, \\ \nonumber
        \pi_3=\{1,2,4,5,6,7\}, \\ \nonumber
        \pi_4=\{2,3\}, \\ \nonumber
        \pi_5=\{3,6,7\}, \\ \nonumber
        \pi_6=\{3,5,7\}, \\ \nonumber
        \pi_7=\{3,5,6,9\}, \\ \nonumber
        \pi_8=\{9\}, \\ \nonumber
        \pi_9=\{7,8\}, \\ \nonumber
        \pi_{10}=\{11\}, \\ \nonumber
        \pi_{11}=\{10\}, \\ \nonumber
        \pi_{12}=\emptyset.
    \end{eqnarray}
\end{exmp}
\begin{exmp}[$d$-faces and Simplicial Complex]
\label{ex:simplicial_complex}
In \cref{fig:CW_example}, the $d$-faces (or the complete subgraphs with $d+1$ nodes), for $d\in \{0,1,2,3\}$ are
\begin{eqnarray}
\nonumber
 C_0=V,\\ \nonumber
 C1=E,\\ \nonumber
 C_2=\{\{2,3,4\},\{3,5,6\},\{3,6,7\},\{3,5,7\},\{5,6,7\},\\ \nonumber
 C_3=\{\{3,5,6,7\}\}. \nonumber
\end{eqnarray}
The simplicial complex is $C=\bigcup_{i=0}^{3}C_i.$
\end{exmp}
\begin{exmp}[Node neighbourhood of a face]\label{ex:node_neigh_of_cell}
In \cref{fig:CW_example}, $\pi_{\{12\}}=\emptyset=\pi_{\{2,3,4\}}$ and $\pi_{3}=\{1,2,4,5,6\}.$ Also, $\pi_{\{3,5,6\}}=\{7\}.$ 
\end{exmp}

\begin{exmp}[Boundary of a cell]\label{ex:boundary_of_cell}
    The boundary of $0$-faces is always an empty set. In \cref{fig:CW_example}, the boundary of the $1$-face $\{7,9\}$ is $$\partial(\{7,9\})=\{\{7\},\{9\}\}.$$ The boundary of the $2$-face $\{2,3,4\}$ is $$\partial(\{2,3,4\})=\{\{2,3\},\{2,4\},\{3,4\}\},$$ and the boundary of the $3$-face $\{3,5,6,7\}$ is $$\partial(\{3,5,6,7\})=\{\{3,5,6\},\{3,5,7\},\{3,6,7\},\{5,6,7\}\}.$$
\end{exmp}
\begin{exmp}[face neighborhood]\label{ex:neigh_exemple}
In \cref{fig:CW_example}, the $0$-face $\{12\}$ has no neighbours, as well as the $1$-face $\{10,11\}$, the $2$-face $\{2,3,4\}$ and the $3$-face $\{3,5,6,7\}.$
The $2$-face $\{3,4\}$ has $6$ neighbours, in which $2$ are transverse ($\{2,3\},\{2,4\}$) and $4$ are parallel ($\{1,3\},\{3,5\},\{3,6\}$ and \{3,7\}). The $2$-face $\{3,5,7\}$ has $3$ neighbours, $\{3,6,7\},\{3,5,6\}$ and $\{5,6,7\}$, in which all of them are transverse.
\end{exmp}
\begin{figure}\label{fig:neigh_condition}
    \centering
    \includegraphics[width=\linewidth]{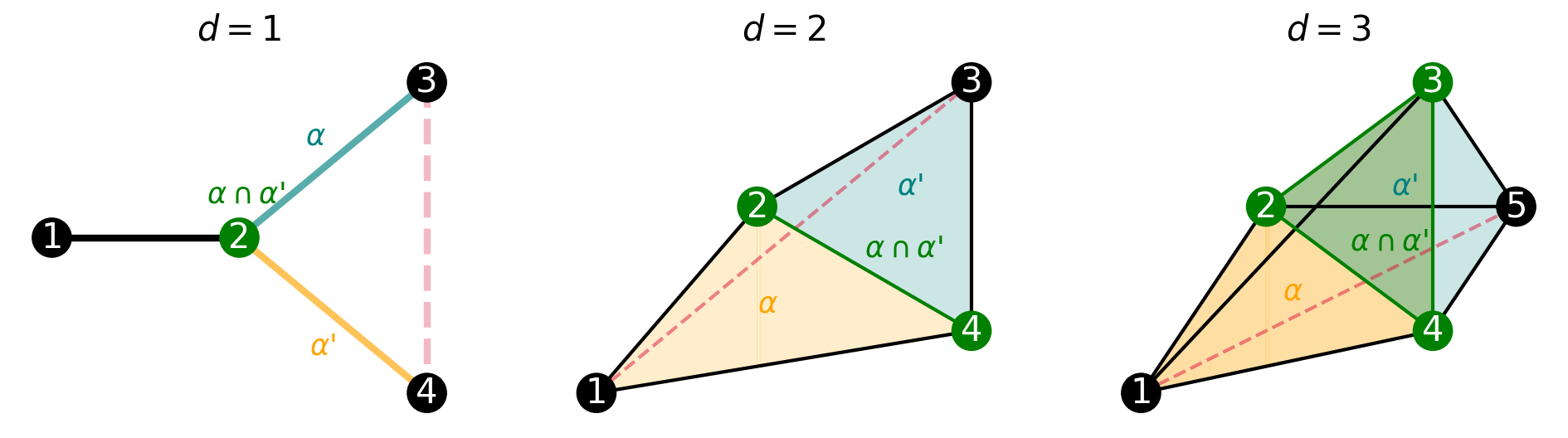}
    \caption{Neighborhood condition depicted, for $d\in\{1,2,3\}$. For all dimensions, the neighbouring faces $\alpha$ and $\alpha'$ are drawn in blue and yellow, respectively, whilst the common boundary ($\alpha\cap\alpha'$) is represented in green. The presence (or absence) of the red dashed edge in the simplicial complex defines whether $\alpha$ and $\alpha'$ are transverse or parallel. In all cases, the intersection $\gamma:=\alpha \cap \alpha '$ is a $(d-1)$-face such that $\gamma<\alpha, \alpha'$, which guarantees that for $\alpha$ and $\alpha'$ to be neighbours, it is sufficient to reach the { face}neighbourhood condition $1.$ is reached. This result is a consequence of \cref{theo:2implies1}.}
\end{figure}
\begin{exmp}[Face neighborhood condition]\label{ex:Neigh_condition}
In \cref{fig:neigh_condition}, we have the characterization of neighborhood of $d$-faces, $\alpha$ (blue) and $\alpha'$ (yellow), for $d\in\{1,2,3\}.$
 Suppose that the red-dashed edge belongs to the graph. For $d=1$, we have that $\alpha=\{2,3\}$ and $\alpha'=\{2,4\}$ are transverse neighbours, once that the $2$-face $\beta=\{2,3,4\}$ is such that $\beta>\alpha,\alpha'$. Analogously, for $d=2$, the $2$-faces $\alpha=\{ 1,2,4\}$ and $\alpha'=\{2,3,4\}$ are transverse once that $\beta=\{1,2,3,4\}$ is a $3$-face such that   $\beta>\alpha,\alpha'$. Also, for $d=3$, $\alpha=\{1,2,3,4\}$, $\alpha'=\{2,3,4,5\}$ are transverse once that $\beta=\{1,2,3,4,5\}$ is a $4$-face such that $\beta>\alpha,\alpha'$. Regardless of the existence of the red-dashed edge in the graph, we have that, for $d\in\{1,2,3\}$, there is a $(d-1)$-cell, $\gamma:=\alpha \cap \alpha'$ such that $\gamma<\alpha,\alpha'$. This fact can be justified by the result of \cref{theo:boundary_charact}, which guarantees that is sufficient to have a common boundary between two $d$-faces so they can be neighbours.
 The dashed-red edge's absence (or presence) guarantees they are parallel (transverse) neighbours.
\end{exmp}
\begin{exmp}[Forman-Ricci curvature]\label{ex:FRC_calculation}
    In \cref{fig:CW_example}, $F_{1}(\{3,4\})=1+2-4=-1,$ $F_{2}(\{3,5,7\})=1+3-0=4.$
\end{exmp}

\subsection{Demonstrations}\label{sec:demonstrations}
Here we are considering $\pi_\alpha$, $\pi_x$ and the concepts of abstract simplicial complexes as defined in \cref{sec:network_background}.
\begin{theorem}\label{theo:C_d}
Let $C_d$ be the set of all $d$-faces of a simplicial complex, as defined in \cref{sec:network_background}. Then,
\begin{eqnarray}\label{eq:C_d_proof}
   C_{d}=\Big\{\alpha:=\{x_0,\hdots,x_d\},\,\,x_i\in\pi_{\alpha\setminus\{x_i\}},\,\forall i \Big\}.
\end{eqnarray}
\end{theorem}
\begin{proof}
    First inclusion will be proven by contradiction. Let $\alpha\in C_d$. Suppose that $\alpha$ is not contained in the right-hand side of the equation above, \textit{i.e.,} there exists $x\in \alpha$ such that $x\notin \pi_{\alpha\setminus \{x\}}.$ It means that exists $y \in \alpha$ such that $y\neq x$ and $x \notin  \pi_y$, which implies that $\alpha$ does not belong to $C_d$.

    Also, let $\alpha$ as on the right-hand side in \eqref{eq:C_d_proof}, from the definition of the set on the right-hand side, in $G$ every pair of vertices in $\alpha$ is connected. Also $|\alpha|=d+1$. Define $V_\alpha=\alpha$ and  $E_\alpha = \{(x_i,x_j)\hspace{1mm}|\hspace{1mm} x_i,x_j\in V_\alpha, x_i\neq x_j\}$ such that $G_\alpha = (V_\alpha,E_\alpha)$.  It is clear that $G_\alpha$ is a complete subgraph with $d+1$ nodes and we can identify $V_\alpha \sim G_\alpha$. It follows that $\alpha \in C_d$. %
\end{proof}

\begin{theorem}\label{theo:boundary_charact}
Let $\alpha,\alpha'\in C_d$, $\alpha\neq \alpha'$. They are neighbours if and only if $\alpha \cap \alpha'$ is the $(d-1)$-face that is common to the boundaries {of $\alpha$ and $\alpha'$}.
    
\end{theorem}
\begin{proof}
The forward implication is proven as follows: Suppose condition 1. of the face neighbourhood is satisfied. There exists $\gamma \in C_{d-1}$ such that $\gamma<\alpha,\alpha'$. In particular, $\gamma\subset \alpha,\alpha'$ and $|\gamma|=d$. It is sufficient to prove that $\gamma=\alpha\cap\alpha '$.
The inclusion $\gamma\subseteq \alpha \cap \alpha'$ is obvious. Suppose, by absurdity, that there exists {$x\in \alpha\cap\alpha'$} such that $x\notin \gamma$. Once the first inclusion is valid, it follows that $\gamma\cup \{x\}\subset (\alpha\cap\alpha')\cup \{x\}=\alpha\cap\alpha'$, which lead us to $|\gamma\cup\{x\}|\leq d$. 

Now, suppose that condition 2. is reached. There exists $\beta \in C_{d+1}$ such that {$\beta > \alpha,\alpha'$. In particular,} $\beta\supset \alpha \cap \alpha':=\gamma.$ First, we need to prove that $|\gamma|$=d. Note that we can write $\beta=\alpha\cup \{x\},$ for some $x \notin \alpha.$ Analogously, we can write $\beta=\alpha '\cup \{y\},$ for some $y\notin \alpha ',$ with $x \neq y.$ By intersecting this equalities we have $\beta=(\alpha\cup\{x\})\cap (\alpha'\cup \{y\})=[\alpha\cap (\alpha'\cup \{y\})]\cup[\{x\}\cap (\alpha'\cup \{y\})]=[\gamma \cup (\alpha\cap\{y\})]\cup [(\{x\}\cap \alpha')\cup (\{x\}\cap\{y\})]=\gamma \cup \{x,y\}.$ Thus, $|\gamma|=|\beta|-2=d.$
Suppose that $\gamma \notin C_{d-1}.$ Then, there exists $y\in \gamma$ such that $y \notin \bigcap_{\substack{x \in \gamma \\ x\neq y}}\pi_x$. However, $\gamma \subset \alpha$, and it would imply that $\alpha\notin C_d.$
The opposite implication of the theorem is obvious.
\end{proof}
\begin{theorem}\label{theo:neighborhood_condition} Let $\alpha, \alpha' \in C_d.$
    Then, $\alpha$ and $\alpha'$ are neighbours if and only if $\left| \alpha \cap \alpha' \right|=d.$
\end{theorem}
\begin{proof}
It follows as a corollary from \cref{theo:boundary_charact}.
This neighbourhood condition can be better elucidated in \cref{ex:Neigh_condition}.  
\end{proof}
\begin{theorem}\label{theo:2implies1}
Let $\alpha,\alpha'\in C_d.$ Then, the neighborhood condition $2.$ implies the condition $1.$,  \textit{i.e.,} $\alpha$ and $\alpha'$ are neighbours if and only if the condition $1.$ of face neighbourhood is reach.
\end{theorem}
\begin{proof}
    It follows as a corollary from \cref{theo:boundary_charact}.
\end{proof}
\begin{theorem}\label{theo:neighb_charact}

    Let $\alpha \in C_d$, and $N_{\alpha}$ as defined originally in \cref{sec:Forman_ricci_curvature}. Then,
    \begin{eqnarray}\label{eq:N_proof}
    N_\alpha=\bigsqcup_{\gamma\in\partial(\alpha)}\,\bigcup_{\substack{x\in\pi_\gamma\neq \emptyset\\x\notin\alpha}}\{\gamma\cup\{x\}\}.
\end{eqnarray}
    
\end{theorem}
\begin{proof}
    {Let $\alpha'\in N_\alpha$. From \cref{theo:boundary_charact}, there exists $\gamma \in C_{d-1}$ such that $\gamma=\alpha\cap\alpha'$ and in particular, $|\gamma|=d.$ Once that $|\alpha|=|\alpha'|=d+1$ and $\gamma\subset\alpha,\alpha'$, there exists $x\in V\setminus\alpha$ such that we can write $\alpha'=\gamma\cup\{x\}.$ Analogously, there exists $y\in V\setminus\alpha'$ such that we can write $\alpha=\gamma\cup\{y\}.$ It is clear that $\pi_{\gamma}\setminus (\alpha\setminus\gamma)=\pi_{\gamma}\setminus \{y\}$. Suppose this set is not empty. We need to prove that $x\in \pi_{\gamma}\setminus (\alpha\setminus\gamma).$ By absurdity, suppose that $x$ is not in this set. By using set complementary property, we have that $x\in (\pi_\gamma)^\mathsf{c}\cup (\alpha\setminus \gamma).$ If $x \in \alpha\setminus\gamma,$ then $\alpha=\alpha '.$ If $x \in \pi_\gamma ^\mathsf{c},$ then there exists $z \in\gamma$ such that $x\notin \pi_\gamma,$ which implies that $\alpha '$ is not a $d$-face.
    }

    {
    The opposite inclusion is proven as follows: Let $\alpha '$ be such that $\alpha'=\gamma\cup \{x\},$ for some $x \in \pi_\gamma \setminus (\alpha\setminus \gamma),$ and some $\gamma \in \partial (\alpha).$ Note that $|\alpha\cap \alpha'|=|\alpha \cap(\gamma\cup \{x\})|=|(\alpha\cap\gamma)\cup (\alpha \cap \{x\})|=|\gamma|=d.$ The result follows from applying \cref{theo:boundary_charact} once proven that $\alpha '\in C_d.$  Indeed, $|\alpha '|=|\gamma \cup \{x\}|=d+1.$ Suppose by absurdity that $\alpha ' \notin C_d.$ Than, by using the the result of \cref{theo:C_d}, there exists $z \in \alpha '$ such that $z\notin \pi_{\alpha '\setminus \{z\}}.$ In case of $z =x$, we have that it contradicts that $x \in \pi_\gamma \setminus \{y\}$. Also, if $z\neq x,$ we would have that $z \in \gamma$ and it contradicts that $\alpha\in C_d.$}
\end{proof}
\begin{theorem}\label{theo:transverse_neigh_decision}
    Let $\alpha \in C_d$, $\gamma \in \partial(\alpha)$ and $x\in \pi_{\gamma},\, x\notin \alpha$. Let $\alpha':=\gamma \cap \{x\}$. Then, $\alpha \in T_{\alpha}$ if and only if $x \in \pi_\alpha$.
\end{theorem}
\begin{proof}
Note that \cref{theo:boundary_charact} and \cref{theo:neighb_charact} guarantees that $\alpha'\in N_{\alpha}$ and that $\gamma=\alpha\cap\alpha'$. Having said that, consider $\beta:=\alpha\cup\alpha'=\gamma \cup \{x,y\},$ for some $y\notin \alpha, \, y\neq x$ If $x \in \pi_\alpha,$ then $\beta \in C_{d+1},$ and $\beta>\alpha,\alpha'.$ Otherwise, there exists $x_0 \in \alpha$ such that $x \not \in \pi_{x_0}$ and then $\beta \notin C_{d+1}.$
\end{proof}
\begin{theorem}\label{theo:set_of_transverse}
    Let $\alpha \in C_d$ and $T_\alpha$ as described in \cref{sec:Forman_ricci_curvature}. Then,
\begin{eqnarray}
T_\alpha=\bigsqcup_{\gamma\in\partial(\alpha)}\,\bigcup_{\substack{x\in\pi_\gamma\neq \emptyset \\ x\in \pi_{\alpha}\neq \emptyset\\x\notin\alpha}}\{\gamma\cup\{x\}\}.
\end{eqnarray}
\end{theorem}
\begin{proof}
It follows as a corollary from \cref{theo:transverse_neigh_decision} and the characterization of the neighbourhood of $\alpha$ in \cref{theo:neighb_charact}.
    
\end{proof}
\begin{theorem}\label{theo:set_of_parallel}
    Let $P_{\alpha}$ as defined in \cref{sec:Forman_ricci_curvature}. Then,
    \begin{equation}
    P_{\alpha}=\bigsqcup_{\gamma\in\partial(\alpha)}\,\bigcup_{\substack{x\in\pi_\gamma\neq \emptyset \\ x\notin \pi_{\alpha}\neq \emptyset\\x\notin\alpha}}\{\gamma\cup\{x\}\}.
\end{equation}
\end{theorem}

\begin{proof}
It also follows as a corollary from \cref{theo:transverse_neigh_decision} and the characterization of \cref{theo:neighb_charact}.
\end{proof}
\begin{theorem}\label{theo:set_of_highn}
    Let $H_\alpha$ as defined in \cref{sec:Forman_ricci_curvature}. Then,
    \begin{eqnarray}
   H_{\alpha}=\bigcup_{x \in \pi_\alpha \neq \emptyset } \{\alpha\cup\{x\} \}.
\end{eqnarray}
\end{theorem}
\begin{proof}
    Let $\beta \in H_\alpha.$ In particular, $\beta \supset \alpha$ implies that exists $x \notin \alpha$ such that we can write $\beta=\alpha \cup \{x\}.$ It is clear that $x\in \pi_\alpha,$ otherwise  $\beta$ would not be a $(d+1)$-face.
    Let $\beta$ be a set of the form $\beta=\alpha \cup \{x\},$ for some $x \in \pi_\alpha.$ It is enough to show that $\beta \in C_{d+1}$, which is immediate once $\left|\beta\right|=d+2.$
\end{proof}
\begin{theorem}\label{theo:T_cardinality}
Let $\alpha \in C_d$. Then,
    \begin{equation}
        \left|T_{\alpha}\right| = (d+1) \left|H_{\alpha}\right|.
    \end{equation}
\end{theorem}
\begin{proof}
    Note that $\pi_\alpha=
    \pi_{\gamma}\cap\pi_{\alpha\setminus\gamma}$, for all $\gamma \in \partial(\alpha).$
    Thus, 
    by taking the cardinality of $T_\alpha$ and $H_\alpha$ using the equalities in \cref{theo:set_of_transverse} and \cref{theo:set_of_highn}, we obtain
    \begin{eqnarray}
        \left|T_{\alpha}\right|=\sum_{\gamma\in\partial(\alpha)}\left|\pi_{\alpha}\right|=(d+1)\left|\pi_\alpha\right|=(d+1)\left|H_{\alpha}\right|.
    \end{eqnarray}
\end{proof}
\textcolor{black}{
\begin{theorem}\label{theo:FRC_range}
    Let $G=(V,E)$. For all $d\geq 1$ and $\alpha\in C_d$, the FRC is bounded by
    \begin{eqnarray}\label{eq:FRC_range}
        2\cdot(d+1)-|V|\leq\mathrm{F}_d(\alpha)\leq |V|
    \end{eqnarray}
\end{theorem}}

\begin{proof}
\textcolor{black}{
    From \eqref{eq:all_cells} and \eqref{eq:forman_d_cells}, it is clear that the minimum of $F_d(\alpha)$ is reached when all neighbours of $\alpha$ are parallel. Analogously, the maximum is reached when all neighbours of $\alpha$ are transverse. Additionally, the maximum number of neighbours of a $d$-face is $|V|-(d+1)$. It follows that the lower and upper bounds of $F_d(\alpha)$ are $2\cdot (d+1)-|V|$ and $|V|$, respectively.
   }
\end{proof}
\begin{theorem}
    Let $d\geq 1$ and $\alpha\in C_d.$ Let also $N_\alpha$ be the set of neighbours of $\alpha$. The FRC is bounded by
    \begin{eqnarray}\label{eq:FRC_neigh_range}
   (d+1)-|N_\alpha| \leq\mathrm{F}_d(\alpha)\leq \left\lfloor\frac{|N_\alpha|}{(d+1)}\right\rfloor+(d+1)
    \end{eqnarray}
    where $\lfloor \, .\, \rfloor$ is the floor of a number.
\end{theorem}
\begin{proof} 
From equation \eqref{eq:FRC_algorithm_v1}, we know that $F_{d}(\alpha) = \frac{|T_\alpha|}{(d+1)} + (d+1) - |P_\alpha|$. Also, recall that the FRC reaches its minimum when all neighbours of $\alpha$ are parallel, i.e., when $P_\alpha = N_\alpha$, and then, $F_{d}(\alpha) = (d+1) - |N_\alpha|$. On the other hand, the FRC reaches its maximum when all neighbours of $\alpha$ are transverse, $T_\alpha = N_\alpha$, where in this case $|N_\alpha|$ is a multiple of $(d+1)$, leading us to $F_{d}(\alpha) = \frac{|N_\alpha|}{(d+1)} + (d+1)$. Otherwise, $|T_\alpha| < |N_\alpha|$ and $F_{d}(\alpha) \leq \left\lfloor \frac{|N_\alpha|}{(d+1)} \right\rfloor + (d+1)$. In the special case where $|N_\alpha| = |V|-(d+1)$ and $N_\alpha = T_\alpha$ (or $N_\alpha = P_\alpha$), the upper (or lower) bound is the same as in \cref{theo:FRC_range}.
\end{proof}
\section*{Acknowledgements}
The authors would like to acknowledge Rodrigo A. Moreira at the Basque Center for Applied Mathematics for his critical review.
This research is supported by the Basque Government through the BERC 2022-2025 program and by the Ministry of Science and Innovation: BCAM Severo Ochoa accreditation CEX2021-001142-S / MICIN / AEI / 10.13039/501100011033. Moreover, the authors acknowledge the financial support received from the IKUR Strategy under the collaboration agreement between the Ikerbasque Foundation and BCAM on behalf of the Department of Education of the Basque Government. SR further acknowledges the RTI2018-093860-B-C21 funded by (AEI/FEDER, UE) and the acronym ``MathNEURO''. We also acknowledge the Elkartek 2023 via grant ONBODY no. KK-2023/00070.
\pagebreak
\bibliographystyle{plain}
\bibliography{references_resub}

\end{document}